\documentclass[a4paper,UKenglish,cleveref, autoref, thm-restate]{lipics-v2021}

\title{Unconditional Time and Space Complexity Lower Bounds for Intersection Non-Emptiness\footnote{Draft as of January 29th, 2026.}} %

\author{Michael Wehar}{Department of Computer Science, Bryn Mawr College, Bryn Mawr, PA}{mwehar@brynmawr.edu}{https://orcid.org/0009-0007-9251-7128}{}

\authorrunning{M. Wehar} %

\Copyright{Michael Wehar} %



\category{} %

\relatedversion{} %

\acknowledgements{I would especially like to thank A. Salamon for years of productive collaboration.  In addition, I am grateful to M. Oliveira for our prior works on intersection non-emptiness.  Also, I very much appreciate the helpful communications with D. Chistikov and N. Rino.}%

\nolinenumbers %

\EventEditors{John Q. Open and Joan R. Access}
\EventNoEds{2}
\EventLongTitle{42nd Conference on Very Important Topics (CVIT 2016)}
\EventShortTitle{CVIT 2016}
\EventAcronym{CVIT}
\EventYear{2016}
\EventDate{December 24--27, 2016}
\EventLocation{Little Whinging, United Kingdom}
\EventLogo{}
\SeriesVolume{42}
\ArticleNo{23}

\newcommand{\DTIME}{\texttt{DTIME}}
\newcommand{\ATIME}{\texttt{ATIME}}
\newcommand{\SPATIME}{\texttt{ATIME}^{O(1)}_{\texttt{SP}}}
\newcommand{\BNSPACE}{\texttt{NSPACE$_b$}}
\newcommand{\DSPACE}{\texttt{DSPACE}}

\newcommand{\NTISP}{\texttt{NTISP}}

\newcommand{\GC}{\texttt{GUESSCHECK}}
\newcommand{\DTIRE}{\texttt{DTIRE}}
\newcommand{\NSPACE}{\texttt{NSPACE}}
\newcommand{\PSPACE}{\texttt{PSPACE}}
\newcommand{\EXPTIME}{\texttt{EXPTIME}}
\newcommand{\NL}{\texttt{NL}}

\newcommand{\PTIME}{\texttt{PTIME}}
\newcommand{\dfaint}{\texttt{DFA{-}INT}}
\newcommand{\twodfa}{\texttt{2DFA{-}NE}}

\newtheorem*{hypothesis}{Hardness Hypothesis}

\newtheorem*{bconjecture}{Binary Space Conjecture}

\begin{document}

\maketitle

\begin{abstract}
We reinvestigate known lower bounds for the Intersection Non-Emptiness Problem for Deterministic Finite Automata (DFA's).
We first strengthen conditional time complexity lower bounds from T. Kasai and S. Iwata (1985) which showed that Intersection Non-Emptiness is not solvable more efficiently
unless there exist more efficient algorithms for non-deterministic logarithmic space ($\NL$).
Next, we apply a recent breakthrough from R. Williams (2025) on the space efficient simulation of deterministic time to show an unconditional $\Omega(\frac{n^2}{\log^3(n) \log\log^2(n)})$
time complexity lower bound for Intersection Non-Emptiness.
Finally, we consider implications that would follow if Intersection Non-Emptiness for a fixed number of DFA's is computationally hard for a fixed polynomial time complexity class.
These implications include $\PTIME \subseteq \DSPACE(n^c)$ for some $c \in \mathbb{N}$ and $\PSPACE = \EXPTIME$.
\end{abstract}

\section{Background}

\subsection{Intersection Non-Emptiness}

The Intersection Non-Emptiness Problem for Deterministic Finite Automata (DFA's) is defined as follows.  Given a finite list of DFA's $D_1$, $D_2$, ..., $D_k$ over an alphabet $\Sigma$, does there exist a string $s \in \Sigma^{*}$ such that $s \in \bigcap_{i \in [k]}{L(D_i)}$?  In other words, is the intersection of the regular languages (associated with the DFA's) non-empty?  We denote this problem by \dfaint.  It became natural to consider \dfaint\ following the work of M. Rabin and D. Scott (1959) \cite{rabin1959} because it was shown that a product automaton could be constructed that recognizes the intersection of any finite set of regular languages.  Using this product construction, one can solve $\dfaint$ in $O(n^{O(k)})$ time where $n$ denotes the total input size and $k$ denotes the number of DFA's.  We will measure the total input size $n$ in terms of the length of the bit string that encodes the entire input.  The standard approach is to imagine a product directed graph and simply search to see if any product final state can be reached from the product initial state.  Furthermore, we can solve $\dfaint$ in $O(k \log(n))$ non-deterministic space because a product state along with a counter can be represented using $O(k \log(n))$ bits of memory.  Then, we simply non-deterministically guess a path from the product initial state to a product final state and reject if the counter surpasses the value $n^k$.
We also consider the $k$-$\dfaint$ problem where $k$, the number of DFA's, is fixed.

The complexity of $\dfaint$ and $k$-$\dfaint$ have been investigated in prior works.  It was shown that $\dfaint$ is a $\PSPACE$-complete problem in D. Kozen (1977) \cite{kozen1977}.  Next, T. Kasai and S. Iwata (1985) \cite{kasai1985} showed that $k$-$\dfaint$ is not solvable more efficiently unless all non-deterministic logspace ($\NL$) problems are solvable in less time.  In particular, it was shown that if $k$-$\dfaint$ is solvable in $O(n^{\frac{k - 2 - \varepsilon}{2}})$ time for some $\varepsilon > 0$, then $\BNSPACE(k \log(n)) \subseteq \DTIME(n^{k - \varepsilon})$ for some $\varepsilon > 0$. Here, $\BNSPACE$ is similar to the common $\NSPACE$, but it essentially measures bits of memory instead of measuring length over larger than binary tape alphabets.  In particular, it measures space complexity relative to two-tape Turing machines where one tape is read-only and the other tape is restricted to a binary alphabet except for a fixed number of special delimiter symbols $\#$.  This is also known as an offline Turing machine.  Later on, the author (2014) \cite{wehar2014} showed that $\dfaint \notin \NSPACE(f(n))$ unconditionally for all $f(n)$ such that $f(n)$ is $o(\frac{n}{\log(n) \log\log(n)})$.  This was proven by demonstrating that a space efficient algorithm for $\dfaint$ could be used to space efficiently simulate all space bounded computations in a way that violates the Non-Deterministic Space Hierarchy Theorem \cite{conl1,conl2,nspace}.
The reason for the $\log \log(n)$ factor in the denominator is because the proof constructs $O(n)$ DFA's with at most $O(\log(n))$ states each.  Furthermore, it takes $\log(n) \log\log(n)$ bits to represent a $\log(n)$ state DFA as a binary string because we are essentially encoding a directed graph with $\log(n)$ vertices and $O(\log(n))$ edges (assuming a fixed finite alphabet).  Additional conditional time complexity lower bounds for $\dfaint$ have been shown in \cite{karakostas2003,wehar2016,fernau2017,oliveira2020}.

\subsection{Structural Complexity}

Structural complexity is concerned with complexity classes and their relationships.  We recall the definitions for a selection of standard complexity classes that we consider within this work.  The standard classes that we consider within this work are polynomial time denoted by $\PTIME = \bigcup_{k\in\mathbb{N}} \DTIME(n^k)$, polynomial space denoted by $\PSPACE = \bigcup_{k\in\mathbb{N}} \DSPACE(n^k)$, exponential time denoted by $\EXPTIME = \bigcup_{k\in\mathbb{N}}\DTIME(2^{n^k})$, and non-deterministic logarithmic space denoted by $\NL = \NSPACE(\log(n))$. %

Within structural complexity theory, it is standard to consider whether one class is included within another.  We recall a selection of known relationships between complexity classes that we consider within this work.
First, the deterministic Time Hierarchy Theorem states that $\DTIME(t_1(n)) \subsetneq \DTIME(t_2(n))$ when $t_2(n)$ is $\omega(t_1(n) \log(t_1(n)))$ and both $t_1(n)$ and $t_2(n)$ are time constructible \cite{time1,time2}.
Next, the Non-Deterministic Space Hierarchy Theorem states that $\NSPACE(s_1(n)) \subsetneq \NSPACE(s_2(n))$ when $s_1(n)$ is $o(s_2(n))$ and $s_2(n)$ is space constructible \cite{trakhtenbrot1964,borodin1972,conl1,conl2,nspace}.
Then, alternating time is included within deterministic space.  In particular, we have $\ATIME(t(n)) \subseteq \DSPACE(t(n))$ \cite{alt1}.  Finally, when space is sublinear, polynomial time is within alternating linear time.
That is, we have $\NTISP(n^k, n^{1 - \varepsilon}) \subseteq \ATIME(n)$ for all $k \in \mathbb{N}$ and $\varepsilon > 0$ \cite{atime1,atime2,atime3,atime4}.

\subsection{Space Efficient Simulation}\label{sec:spaceefficient}

Time bounded computations can be space efficiently simulated.  In particular, it was shown in \cite{hopcroft1} that $\DTIME(t(n)) \subseteq \DSPACE(\frac{t(n)}{\log(t(n))})$.  Then, a more space efficient simulation was demonstrated in \cite{hopcroft2} when the time bounded machine has a read tape and only one read / write tape (also known as an offline Turing machine).  We use $\DTIME_1$ to measure time relative to these restricted machines with only one tape that can both read and write.
In particular, the improved simulation from \cite{hopcroft2} demonstrated that $\DTIME_1(t(n)) \subseteq \DSPACE(\sqrt{t(n)\log(t(n))})$.  A more recent breakthrough was made demonstrating that for the more standard multitape Turing machine model, we also have $\DTIME(t(n)) \subseteq \DSPACE(\sqrt{t(n)\log(t(n))})$ \cite{williams2025}.  This result is based on the space efficient algorithm for Tree Evaluation from \cite{treeeval}.

\section{Introduction}

\subsection{Motivation}
There are two primary motivations for this work.  Firstly, we present modest improvements to the known lower bounds for the Intersection Non-Emptiness for DFA's Problem ($\dfaint$) which is a core problem within the Reachability and Automata Theory research communities.  Secondly, few natural problems in computer science theory have known unconditional time complexity lower bounds.  We demonstrate that Intersection Non-Emptiness for DFA's is such an example.  The previously known examples tend to be $\PTIME$-hard and $\EXPTIME$-hard problems which are therefore at least as hard as simulating time bounded Turing machines \cite{lowerbounds1,lowerbounds2}.
In particular, combinatorial problems related to pebbling games were shown to have polynomial time lower bounds in \cite{kasai1984}.  Also, for some $c > 0$, intersection non-emptiness for $k$ DFA's and one PDA was shown to have an $\Omega(n^{c k})$ time lower bound for all fixed $k \in \mathbb{N}$ \cite{wehar2015}.

We will see that the classic result on the space efficient simulation of time bounded computations (from \cite{hopcroft1}) implies that $\PSPACE$-hard problems may also have
time complexity lower bounds, especially with the recent improvement of this simulation from \cite{williams2025}.
This is how we will demonstrate the time complexity lower bounds for $\dfaint$ in Section \ref{sec:timelower}.
We hope that this work will further lead to unconditional time complexity lower bounds for other $\PSPACE$-hard problems as well.

\subsection{Our Contribution}
First, in Theorem \ref{conditionallower1}, we improve the conditional lower bounds of \cite{kasai1985} showing that $\dfaint \notin \DTIME(n^{k-\varepsilon})$ for all $\varepsilon > 0$ unless $\BNSPACE(k \log(n)) \subseteq \DTIME(n^{k - \varepsilon})$ for some $\varepsilon > 0$.  Next, in Theorem \ref{spacelower1},
we restate the unconditional space complexity lower bounds of \cite{wehar2014}, showing that $\dfaint$ is not solvable non-deterministically in $o(\frac{n}{\log(n)\log\log(n)})$ space.
Then, we show in Theorem \ref{timelower2}, using the recent breakthrough from \cite{williams2025}, that $\dfaint$ is not solvable in $o(\frac{n^2}{\log^3(n)\log\log^2(n)})$ deterministic time.
Finally, in Theorem \ref{hardnesshypo2}, we consider the implications of a hardness hypothesis for the $k$-$\dfaint$ problems.
That is, we show that if $\dfaint$ for a fixed number of DFA's is $\DTIME(n^{\beta + \varepsilon})$-hard under $O(n^{\beta})$ time reductions for some
$\beta \geq 1$ and $\varepsilon > 0$, then $\PTIME \subseteq \DSPACE(n^{\beta})$.  Furthermore, we show in Corollary \ref{hardnesshypo3} that this hardness hypothesis also implies that $\PSPACE = \EXPTIME$.

\section{Conditional Time Lower Bound}

Let $k \in \mathbb{N}$ be given.  We denote by $k$-$\dfaint$ the $\dfaint$ problem where the number of DFA's is restricted to be $k$.
In \cite{kasai1985}, the following conjecture was made about the relationship between non-deterministic logarithmic binary space and deterministic polynomial time.

\begin{bconjecture}{(\cite{kasai1985})}
For all $k \in \mathbb{N}$ and $\varepsilon > 0$, $$\BNSPACE(k \log(n)) \nsubseteq \DTIME(n^{k - \varepsilon}).$$
\end{bconjecture}

Assuming that this conjecture is true, it was demonstrated that $k$-$\dfaint$ is not solvable in $O(n^{\frac{k - 2 - \varepsilon}{2}})$ deterministic time for all $k > 8$ and $\varepsilon > 0$ \cite{kasai1985}.
In this section, we first strengthen the preceding result by introducing a tighter reduction in Proposition \ref{prop:conditionallower}.
Then, we introduce an amplification technique in Lemma \ref{lem:amp} to further tighten the result.\footnote{A preliminary form of this result was intended by the author to be included in \cite{oliveira2020}, but it did not make the final version.}

\begin{proposition}\label{prop:conditionallower}
The Binary Space Conjecture implies $(k+1)$-$\dfaint \notin \DTIME(n^{k-\varepsilon})$ for all $k \in \mathbb{N}$ and $\varepsilon > 0$.
\end{proposition}
\begin{proof}
We proceed with a variation on the standard reduction from \cite{kozen1977,karakostas2003,oliveira2020} where it is shown how DFA's can be constructed to verify the correctness of a space bounded computation.  Let $k \in \mathbb{N}$ be given.  Let a language $L \in \BNSPACE(k \log(n))$ and a non-deterministic Turing machine $M$ such that $M$ decides $L$ using at most $k \log(n)$ space be given.  Let an input string $s$ of length $n$ be given.  We construct $k + 1$ DFA's as follows.  The DFA's read in a sequence of $6$-tuples $(q, r_0, r_1, m_0, m_1, w)$ encoded in binary such that $q$ represents the current state. We use $r_0$ and $r_1$ to represent what is currently read on the input and work tapes of $M$ on input $s$, respectively.  Similarly, we use $m_0$ and $m_1$ to represent which direction the input and work tapes move, respectively.  We use $w$ to represent what will be written to the work tape.  One DFA $\mathcal{D}_0$ will keep track of the current state, the input tape head position, and verify the correctness of each $q$, $r_0$, $m_0$, $m_1$, and $w$.  For each $i > 0$, the DFA $\mathcal{D}_i$ will keep track of the worktape position and verify the correctness of each $r_1$ when the worktape position is between $(i-1) \log(n)$ and $i \log(n) - 1$.  In other words, we broke the worktape into $k$ blocks each of length $\log(n)$ where each $\mathcal{D}_i$ stores the contents of the $i$th block.  These blocks store binary strings representing the tape content with a fixed number of delimiter $\#$ symbols.  Together, the DFA's will verify that the sequence corresponds with a valid and accepting computation of $M$ on input $s$.  The DFA $\mathcal{D}_0$ will have $\tilde{O}(\vert M \vert^2 n)$ states while each $\mathcal{D}_i$ will have $\tilde{O}(n)$ states.  Because $\vert M \vert$ is a constant, we have constructed $k+1$ DFA's each with $\tilde{O}(n)$ states.  Now, we proceed with the contrapositive.  If $(k+1)$-$\dfaint \in \DTIME(n^{k-\varepsilon})$ for some $k \in \mathbb{N}$ and $\varepsilon > 0$, then by the reduction, we have $\BNSPACE(k\log(n)) \subseteq \DTIME(\tilde{O}(n^{k - \varepsilon}))$.
Therefore, the Binary Space Hypothesis would not hold.
\end{proof}

The following amplification lemma works because $\dfaint$ has a self-reducibility property where for all $d$ and $k \in \mathbb{N}$, we have $(d\cdot k)$-$\dfaint$ on DFA's of size $n$ is reducible to $k$-$\dfaint$ on DFA's of size $n^d$ by applying the classic Cartesian product construction from \cite{rabin1959}.

\begin{lemma}\label{lem:amp}
Let $c \in \mathbb{N}$ be given.  If $(k+c)$-$\dfaint \notin \DTIME(n^{k - \varepsilon})$ for all $k \in \mathbb{N}$ and $\varepsilon > 0$, then $k$-$\dfaint \notin \DTIME(n^{k-\varepsilon})$ for all $k \in \mathbb{N}$ and $\varepsilon > 0$.
\end{lemma}
\begin{proof}
Let $c \in \mathbb{N}$ be given.  Suppose that $(k+c)$-$\dfaint \notin \DTIME(n^{k - \varepsilon})$ for all $k \in \mathbb{N}$ and $\varepsilon > 0$.  Suppose for sake of contradiction that $k$-$\dfaint \in \DTIME(n^{k-\varepsilon})$ for some $k$ and $\varepsilon > 0$.
Consider $\alpha = \lceil \frac{c + 1}{\varepsilon} \rceil$.
By applying the Cartesian product construction \cite{rabin1959}, we have
$$(\alpha \cdot k)\text{-}\dfaint \in \DTIME(n^{\alpha \cdot (k - \varepsilon)}) \subseteq \DTIME(n^{\alpha \cdot k - c - 1}).$$
Now, by assigning $r = \alpha \cdot k - c$, we get $(r + c)$-$\dfaint \in \DTIME(n^{r - 1})$.
Therefore, we have contradicted the assumption and obtained the desired result.
\end{proof}

\begin{theorem}\label{conditionallower1}
The Binary Space Conjecture implies $k$-$\dfaint \notin \DTIME(n^{k-\varepsilon})$ for all $k \in \mathbb{N}$ and $\varepsilon > 0$.
\end{theorem}
\begin{proof}
Combine the results of Proposition \ref{prop:conditionallower} and Lemma \ref{lem:amp} (for $c = 1$).
\end{proof}

\section{Unconditional Space Lower Bound}

From \cite{wehar2014}, we know that $\dfaint \notin \NSPACE(f(n))$ for all $f(n)$ such that $f(n)$ is $o(\frac{n}{\log(n)\log\log(n)})$.
This result is demonstrated by reducing the simulation of an $n$-space bounded Turing machine to solving $\dfaint$ for $O(n)$ DFA's each with at most $\log(n)$ states.
When encoded in binary, the $\log(n)$ state DFA's will be represented by bit strings of length $\log(n)\log\log(n)$.
The original construction from \cite{wehar2014} is missing a few important elements.
In particular, the informative configurations should also include the write bit and the move directions of both tapes to ensure that all DFA's agree on which non-deterministic transition to take.
Below we provide a refined presentation of this argument from \cite{wehar2014}.

\begin{theorem}{(\cite{wehar2014})}\label{spacelower1}
$\dfaint \notin \NSPACE(f(n))$ for all $f(n)$ such that $f(n)$ is $o(\frac{n}{\log(n)\log\log(n)})$.
\end{theorem}

\begin{proof}
Let a language $L \in \NSPACE(n)$ and an input string $s$ of length $n$ be given.  Therefore, there exists a non-deterministic Turing machine $M$ that decides $L$ using at most $O(n)$ space.
We construct $O(n)$ DFA's each with at most $O(\log(n))$ states such that $M$ accepts $s$ if and only if the DFA's have a non-empty intersection.
The DFA's read a sequence of $8$-tuples $(q, h_0, h_1, r_0, r_1, m_0, m_1, w)$ encoded in binary such that $q$ represents the current state.
The values $h_0$ and $h_1$ were not included in Proposition \ref{prop:conditionallower}.  We use $h_0$ and $h_1$ to represent the current input and work tape positions, respectively.
We use $r_0$ and $r_1$ to represent what is currently read on the input and work tapes of $M$ on input $s$, respectively.
We use $m_0$ and $m_1$ to represent which direction the input and work tapes move, respectively.  We use $w$ to represent what will be written to the work tape.

We break the DFA's into four groups.
In the first group, we have one DFA to keep track of the current state and verify that $m_0$, $m_1$, $w$, and the next tuple's $q$ all correspond with a valid transition from the current state.
This DFA has $O(\vert M \vert^2)$ states.
In the second group, we have $O(n)$ DFA's for verifying the correctness of the $h_0$ and $r_0$ values.  Each DFA is assigned an input tape position.
For this position, the DFA will verify that the read bit $r_0$ is correct.  It will also check that the following tuple's $h_0$ value is one more or one less depending on $m_0$'s value.
Each DFA has $O(\log(n))$ states because it essentially just needs to read and check that $O(1)$ many bit strings of length $O(\log(n))$ appropriately match.
In the third group, we have $O(n)$ DFA's for verifying the correctness of the $h_1$ and $r_1$ values.  Each DFA has $O(\log(n))$ states and is constructed similar to the second group.
In the fourth group, we have $O(n)$ DFA's for keeping track of the worktape contents.  Each DFA is assigned a worktape position.
For this position, the DFA keeps track of the bit currently stored there, verifying that the $r_1$ bit matches when the DFA's position is $h_1$, and updates the stored
worktape content based on $w$ when appropriate.
Each DFA has $O(\log(n))$ states because it again just needs to read and check that $O(1)$ many bit strings of length $O(\log(n))$ appropriately match.

It remains to select the initial states appropriately with all of the tape heads moved to the left and with the worktape initially configured with $0$'s.
In addition, transitions to dead states are made when the checks that the DFA's are assigned to fail.  Also, for the first group DFA, final states are selected for when $q$ is final in $M$
and the other DFA's have all states as final that aren't dead states.
By combining this reduction with the Non-Deterministic Space Hierarchy Theorem from \cite{conl1,conl2,nspace}, we obtain the desired result.
\end{proof}

\begin{remark}
We suggest that there is an optimization for constructing the second and third group DFA's.  In particular, we can break the tapes into contiguous blocks of size $O(\log\log(n))$.
Then, we can accomplish their verification tasks using $O(\frac{n}{\log\log(n)})$ DFA's each with $O(\log(n))$ states.
\end{remark}

\begin{remark}
It was suggested by M. Oliveira (and shown in manuscript \cite{oliveira2020manuscript}) that the DFA's obtained from the construction of Theorem \ref{spacelower1} can be represented
in a succinct way that doesn't involve encoding the entire graph structure.
In particular, each DFA can be succinctly represented using only $\log(n)$ bits instead of $\log(n)\log\log(n)$ bits.
Furthermore, $\dfaint$ over succinctly represented DFA's is not non-deterministically solvable in $o(\frac{n}{\log(n)})$ space.
\end{remark}

\section{Unconditional Time Lower Bounds}\label{sec:timelower}

In this section, we combine Theorem \ref{spacelower1} from the preceding section with the space efficient simulation of deterministic time bounded multitape Turing machines
from \cite{williams2025} to show that $\dfaint \notin \DTIME(n^{2 - \varepsilon})$ for all $\varepsilon > 0$.  In other words, $\dfaint$ is not solvable deterministically in
strongly subquadratic time.  We proceed with the following lemma.

\begin{lemma}\label{lem:lower}
If $\dfaint \notin \DSPACE(f(n))$, then $\dfaint \notin \DTIME(\frac{f(n)^2}{\log(f(n))})$.
\end{lemma}
\begin{proof}
We will prove the contrapositive.  Suppose that $\dfaint \in \DTIME(\frac{f(n)^2}{\log(f(n))})$.
By applying the space efficient simulation\footnote{We refer the reader to the discussion from Section \ref{sec:spaceefficient}.} of time bounded computations from \cite{williams2025}, we have that $\DTIME(\frac{f(n)^2}{\log(f(n))}) \subseteq \DSPACE(f(n))$.
Therefore, it follows that $\dfaint \in \DSPACE(f(n))$.
\end{proof}

\begin{theorem}\label{timelower2}
$\dfaint \notin \DTIME(t(n))$ for all $t(n)$ such that $t(n)$ is $o(\frac{n^2}{\log^3(n) \log\log^2(n)})$.
\end{theorem}
\begin{proof}
Combine Theorem \ref{spacelower1} with Lemma \ref{lem:lower} where $f(n)$ is $o(\frac{n}{\log(n)\log\log(n)})$.
\end{proof}

\begin{corollary}\label{timelower3}
$\dfaint \notin \DTIME(n^{2-\varepsilon})$ for all $\varepsilon > 0$.
\end{corollary}
\begin{proof}
Follows directly from Theorem \ref{timelower2}.
\end{proof}

\begin{remark}
It is worth noting that a restricted variant of this result can still be obtained even without the recent breakthrough on space efficient simulation.
In particular, we can apply the classic result of \cite{hopcroft2} on space efficient simulation of deterministic time bounded one-tape offline Turing machines
to obtain that $\dfaint \notin \DTIME_1(n^{2 - \varepsilon})$ for all $\varepsilon > 0$.
\end{remark}

The preceding lower bounds can be improved by allowing for a limited amount of non-determinism.  To demonstrate this, we define a complexity measure $\GC$
which is adapted from the guess and check complexity classes in \cite{guesscheck}.  Before proceeding to our definition, we first define basic properties of functions.
We say that a function $f: \mathbb{N} \rightarrow \mathbb{N}$
is non-decreasing if $n_1 \leq n_2$ implies $f(n_1) \leq f(n_2)$ for all $n_1$, $n_2 \in \mathbb{N}$.  Similar to \cite{subadditive}, we say that $f(n)$ is subadditive if
$f(n_1 + n_2) \leq f(n_1) + f(n_2)$ for all $n_1$, $n_2 \in \mathbb{N}$.  We say that $f(n)$ is sublinear if $f(n)$ is $o(n)$.  As introduced in \cite{constructible}, we say that $f(n)$ is space-constructible
if $f(n)$ can be computed in $O(f(n))$-space by a Turing transducer.  Finally, we say that $f(n)$ is a \textbf{qualified witness} if it is non-decreasing, subadditive, sublinear, and space-constructible.

Now, given a language $L$ and a qualified witness $g(n)$, a verifier for $L$ with witness length $g(n)$
is a deterministic Turing machine $M$ satisfying the following.  For every string $x$ of length $n$, $x \in L$ if and only if there exists a bit string $w$ of
length at most $g(n)$ such that $x \# w$ is accepted by $M$. %
Furthermore, we write $L \in \GC(g(n), t(n))$ if there exists a verifier for $L$ with witness length $g(n)$ that runs for at most
$O(t(n))$ steps.  Notice that for the witness length bound $g(n)$, the input $n$ is the length of $x$.  However, for the time bound $t(n)$, the
input $n$ is the length of $x \# w$. In other words, the time is measured in terms of the combined input and witness lengths.
We proceed by making an improvement to Lemma \ref{lem:lower}.

\begin{lemma}\label{lem:guesscheck}
Let a qualified witness $g(n)$ be given.  If $\dfaint \notin \DSPACE(g(n))$, then
$\dfaint \notin \GC(g(n), \frac{g(n)^2}{\log(g(n))})$.
\end{lemma}
\begin{proof}
We will prove the contrapositive.  Suppose $\dfaint \in \GC(g(n), \frac{g(n)^2}{\log(g(n))})$.
Hence, $\dfaint$ has a verifier $M$ with witness length $g(n)$ that runs for $O(\frac{g(n)^2}{\log(g(n))})$ steps.
By applying the space efficient simulation of time bounded computations from \cite{williams2025}, we get a new verifier $M^{\prime}$ such that
$L(M) = L(M^{\prime})$ and $M^{\prime}$ uses at most $O(g(n))$ space.  We can now convert $M^{\prime}$ into a deterministic machine $M^{\prime\prime}$
that on an input of length $n$ runs the verifier $M^{\prime}$ on all witnesses of length $m$ such that $m \leq g(n)$.  If a witness $w$
is found such that $x \# w$ is accepted by $M^{\prime}$, then accept $x$.  Otherwise, if no witness is found, then reject $x$.  Therefore, $M^{\prime\prime}$
decides $\dfaint$.  Furthermore, we can do this simulation using at most $O(g(n) + g(n + g(n) + 1))$ space because it only takes $O(g(n))$ space to write down
the witness and then $O(g(n + g(n) + 1))$ space to simulate $M^{\prime}$.  Because $g(n)$ is subadditive, $g(n + g(n) + 1) \leq g(n) + g(g(n)) + g(1)$.
Because $g(n)$ is sublinear, $g(n)$ is $o(n)$.  Combining with $g(n)$ being non-decreasing implies that $g(g(n))$ is $O(g(n))$.  Therefore, the total
space is $O(g(n))$.  It follows that $\dfaint \in \DSPACE(g(n))$.
\end{proof}

\begin{corollary}\label{timelower4}
Let a qualified witness $g(n)$ be given.  If $g(n)$ is $o(\frac{n}{\log(n)\log\log(n)})$, then we have $\dfaint \notin \GC(g(n), \frac{g(n)^2}{\log(g(n))})$.
Furthermore, $\dfaint \notin \GC(n^{1-\varepsilon}, n^{2-\varepsilon})$ for all $\varepsilon > 0$.
\end{corollary}
\begin{proof}
We obtain the desired result by applying Lemma \ref{lem:guesscheck} similar to how Lemma \ref{lem:lower} was applied to prove Theorem \ref{timelower2}
and Corollary \ref{timelower3}.
\end{proof}

We consider a related problem called non-emptiness for two-way deterministic finite automata.  We denote this problem by $\twodfa$.  A two-way
deterministic finite automaton (2DFA) is an automaton that can switch between moving left or right while reading the input string.  The input string also
has special left and right end-marker characters so that the automaton can detect when it has reached the end of the input string.  We now define the
$\twodfa$ problem as follows.  Given a 2DFA, does there exist a string that is accepted by the automaton?  Similar to $\dfaint$, the $\twodfa$ problem is
well known to be $\PSPACE$-complete \cite{hunt1973,galil1976}.  In the following, we see that the lower bounds from Theorem \ref{timelower2}, Corollary
\ref{timelower3}, and Corollary \ref{timelower4} also apply to $\twodfa$.

\begin{corollary}\label{timelower5}
The time complexity lower bounds for $\dfaint$ from this section also apply to the $\twodfa$ problem.  Furthermore, $\twodfa \notin \GC(n^{1-\varepsilon}, n^{2-\varepsilon})$ for all $\varepsilon > 0$.
\end{corollary}
\begin{proof}
From the argument in Theorem \ref{spacelower1}, we know that simulating an $O(n)$-space bounded non-deterministic machine can be reduced to solving $\dfaint$
for an instance with $O(n)$ DFA's each with at most $O(\log(n))$ states.  We further reduce intersecting $O(n)$ DFA's each with at most $O(\log(n))$ states
to solving non-emptiness for a 2DFA with $O(n \cdot \log(n))$ states which is encoded as a binary input string of length $O(n \cdot \log(n) \cdot \log\log(n))$.
By combining with the Non-Deterministic Space Hierarchy Theorem from \cite{conl1,conl2,nspace}, we get that $\twodfa$ is not non-deterministically solvable in
$o(\frac{n}{\log(n)\log\log(n)})$ space.  Therefore, by applying similar arguments to Theorem \ref{timelower2}, Corollary \ref{timelower3}, and
Corollary \ref{timelower4}, we obtain the desired result.

All that remains is to present the reduction from $\dfaint$ to $\twodfa$.  Given $O(n)$ DFA's each with at most $O(\log(n))$ states, we construct an associated 2DFA
by stitching together the one-way DFA's.  Essentially, the 2DFA starts by simulating the first DFA on the input string
moving to the right.  Once the right end-marker is reached, if the first DFA has reached a final state, then move all the way back to the
beginning of the input string.  Once the left end-marker is reached, start simulating the second DFA on the input string moving to the right.  We repeat
this process for each DFA.  If we successfully simulate all of the DFA's and the last DFA's final state is reached, then we accept.  If we ever don't reach
a final state when at the right end-marker, then we go to a dead state and reject.  The resulting 2DFA will have $O(n \cdot \log(n))$ states and can
be encoded as a binary input string of length $O(n \cdot \log(n) \cdot \log\log(n))$.  In addition, this 2DFA is efficient to construct as required.
\end{proof}

\section{Hardness Hypothesis}\label{sec:hardhypo}

In the preceding section, we proved an unconditional time complexity lower bound for the $\dfaint$ problem.
We did so by first showing that every language in $\NSPACE(\frac{n}{\log(n)\log\log(n)})$ is efficiently reducible to $\dfaint$.
Then, we observed that $\DTIME(n^{2-\varepsilon}) \subseteq \NSPACE(\frac{n}{\log(n)\log\log(n)})$ for all $\varepsilon > 0$ by applying \cite{williams2025}.
Therefore, every language in $\DTIME(n^{2-\varepsilon})$ is efficiently reducible to $\dfaint$.
The reductions run in nearly linear time because a Turing machine for a language in $\DTIME(n^{2-\varepsilon})$ is fixed
meaning that we basically just need to take the input string and hard code it into the DFA's from the construction.
Furthermore, these DFA's are almost entirely determined by the Turing machine and the input size.

From the preceding, we observe that $\dfaint$ is $\DTIME(n^{2 - \varepsilon})$-hard under nearly linear time reductions. %
In this section, we consider the possibility that $k$-$\dfaint$ is hard for a superlinear fixed polynomial time complexity class for some $k \in \mathbb{N}$.
If we consider hardness under logspace reductions, then this would clearly be difficult to prove because it would imply that $\NL = \PTIME$ since $k$-$\dfaint \in \NL$.
However, what if we consider hardness under polynomial time bounded reductions?  This leads us to the following hardness hypothesis.

\begin{hypothesis}
$k$-$\dfaint$ is $\DTIME(n^{\beta + \varepsilon})$-hard under $O(n^{\beta})$ time reductions for some $k \in \mathbb{N}$, $\beta \geq 1$, and $\varepsilon > 0$.
\end{hypothesis}

In the following, we demonstrate how the hardness hypothesis has significant implications within structural complexity theory.  In particular, assuming the hypothesis,
we can repeatedly apply it to obtain that $\PTIME \subseteq \DSPACE(n^{\beta})$ and $\PSPACE = \EXPTIME$.  In other words, showing unconditional time complexity lower bounds
in this way for the $k$-$\dfaint$ problems would be at least as difficult as resolving major open problems in structural complexity theory
such as $\PTIME$ vs $\PSPACE$ and $\PSPACE$ vs $\EXPTIME$.  We proceed by making the following initial observation about the hardness hypothesis.

\begin{proposition}\label{hardnesshypo4}
The Hardness Hypothesis implies that $\NL \nsubseteq \DTIME(n^{\alpha})$ for some $\alpha > 1$.
\end{proposition}

\begin{proof}
Suppose that the Hardness Hypothesis holds.
Hence, $k$-$\dfaint$ is $\DTIME(n^{\beta + \varepsilon})$-hard under $O(n^{\beta})$ time reductions for some $k \in \mathbb{N}$, $\beta \geq 1$, and $\varepsilon > 0$.
We cannot have $\alpha$ such that $1 < \alpha < \frac{\beta + \varepsilon}{\beta}$ and $k$-$\dfaint \in \DTIME(n^{\alpha})$.  Otherwise, we would get that
$\DTIME(n^{\beta + \varepsilon}) \subseteq \DTIME(n^{\alpha \cdot \beta})$, but this violates the Time Hierarchy
Theorem \cite{time1,time2} because $\alpha \cdot \beta < \beta + \varepsilon$.  Since we also have that $k$-$\dfaint \in \NL$, it follows that $\NL \nsubseteq \DTIME(n^{\alpha})$ for all $\alpha$ such that
$1 < \alpha < \frac{\beta + \varepsilon}{\beta}$.
\end{proof}

In the following, we carefully consider alternating time bounded computations with a restricted amount of alternations and
total witness length.  We denote by $\ATIME^{a(n)}_{w(n)}(t(n))$ the class of languages that can be decided by an alternating
Turing machine in $O(t(n))$ time with at most $a(n)$ alternations and at most $w(n)$ combined binary witness length across all quantifiers.
Furthermore, we write $\texttt{SP}$ as an abbreviation for subpolynomial.  That is,
$\SPATIME(t(n)) = \bigcap_{\varepsilon > 0} \ATIME^{O(1)}_{n^{\varepsilon}}(t(n))$.
We obtain the following proposition by carefully analyzing the alternations and combined witness length from the known proof of
$\NTISP(n^k, n^{1 - \varepsilon}) \subseteq \ATIME(n)$ for all
$k \in \mathbb{N}$ and $\varepsilon > 0$ \cite{atime1,atime2,atime3,atime4}.

\begin{proposition}\label{prop:atime1}
Let $k \in \mathbb{N}$ and $\varepsilon > 0$ be given.  We have $\NTISP(n^k, n^{1 - \varepsilon}) \subseteq \ATIME^{O(1)}_{o(n)}(n)$.
Furthermore, for any subpolynomial function $g(n)$, we have $\NTISP(n^k, g(n)) \subseteq \SPATIME(n)$.
\end{proposition}

\begin{proof}
Let $k \in \mathbb{N}$ and $\varepsilon > 0$ be given.  Let a language $L \in \NTISP(n^k, n^{1 - \varepsilon})$ be given.  Consider a non-deterministic
Turing machine $M$ that decides $L$ in at most $O(n^{k})$ time using at most $O(n^{1 - \varepsilon})$ space.  We proceed similar to
\cite{atime1,atime2,atime3,atime4} where we use existential quantifiers ($\exists$) to guess tape configurations of size
$O(n^{1 - \varepsilon})$ and then universal quantifiers ($\forall$) to check that for every adjacent pair of configurations, there exists a valid
computation leading from the earlier configuration to the later configuration within the pair.  Now, consider making the existential quantifiers guess $n^{\alpha}$
many configurations for $\alpha > 0$.  Consider repeatedly alternating between quantifiers in this way until configurations are only one computational step apart.
This results in $\frac{2 k}{\alpha}$ alternations and $O(\frac{k}{\alpha} \cdot n^{\alpha} \cdot n^{1 - \varepsilon})$ combined binary witness length.
Since $k$ and $\alpha$ are constants that do not depend on $n$, we have $O(1)$ alternations and $O(n^{1 - \varepsilon + \alpha})$ witness length.
If we pick $\alpha = \frac{\varepsilon}{2}$, then we have $O(n^{1 - \frac{\varepsilon}{2}}) = o(n)$ witness length.
Therefore, $\NTISP(n^k, n^{1 - \varepsilon}) \subseteq \ATIME^{O(1)}_{o(n)}(n)$ for all $k \in \mathbb{N}$ and $\varepsilon > 0$.
If the Turing machine instead uses $g(n)$ spaces where $g(n)$ is subpolynomial, then this same construction yields an alternating
simulation with $O(1)$ alternations and $O(g(n) \cdot n^{\alpha})$ witness length.  Since this works for all $\alpha > 0$, we obtain
$\NTISP(n^k, g(n)) \subseteq \SPATIME(n)$.
\end{proof}

Next, assuming the Hardness Hypothesis, we repeatedly apply it similar to a speed-up argument to obtain implications about the relationships between time
and space complexity.  To do so, we define a complexity measure that we denote by $\DTIRE$ (Time-Reducible).  In particular, we define
$\DTIRE(t(n), s(n))$ to be the set of languages that can be reduced in $O(t(n))$ time to a language in $\NSPACE(s(n))$.

\begin{theorem}\label{hardnesshypo2}
Let $\beta \geq 1$ be given.  If $k$-$\dfaint$ is $\DTIME(n^{\beta + \varepsilon})$-hard under $O(n^{\beta})$ time reductions for some $k \in \mathbb{N}$ and $\varepsilon > 0$, then $\PTIME \subseteq \SPATIME(n^{\beta}) \subseteq \DSPACE(n^{\beta})$.  Moreover, the Hardness Hypothesis implies that $\PTIME \subseteq \DSPACE(n^{c})$ for some $c > 0$.
\end{theorem}

\begin{proof}
Let $k \in \mathbb{N}$, $\beta \geq 1$, and $\varepsilon > 0$ be given.  Suppose that $k$-$\dfaint$ is $\DTIME(n^{\beta + \varepsilon})$-hard under $O(n^{\beta})$ time reductions.
Therefore, every language in $\DTIME(n^{\beta + \varepsilon})$ can be reduced in $O(n^{\beta})$ time to $k$-$\dfaint$.  Since $k$-$\dfaint \in \NL$, we have that every language in
$\DTIME(n^{\beta + \varepsilon})$ can be reduced in $O(n^{\beta})$ time to a language in $\NL = \NSPACE(\log(n))$.  Hence, $\DTIME(n^{\beta + \varepsilon}) \subseteq \DTIRE(n^{\beta}, \log(n))$.
Therefore, we have that $$\DTIME(p(n)^{\beta + \varepsilon}) \subseteq \DTIRE(p(n)^{\beta}, \log(n))$$ by a straightforward padding argument.
We can now apply $$\NL \subseteq \bigcup_{k \in \mathbb{N}}\NTISP(n^k, \log(n)) \subseteq \SPATIME(n)$$ by \mbox{Proposition \ref{prop:atime1}} to get that
$$\DTIME(p(n)^{\beta + \varepsilon}) \subseteq \DTIRE(p(n)^{\beta}, \log(n)) \subseteq \SPATIME(p(n)^{\beta}).$$

Let $c \in \mathbb{N}$ be given.  Let $h(d) = (\frac{\beta}{\beta + \varepsilon})^d$.  We proceed by showing that $\DTIME(n^c) \subseteq \SPATIME(n^{\beta} + n^{c \cdot h(d)})$ for every $d \in \mathbb{N} \cup \{0\}$.
The base case ($d = 0$) trivially holds since $\DTIME(n^c) \subseteq \SPATIME(n^c)$.  For the inductive step, suppose that $$\DTIME(n^c) \subseteq \SPATIME(n^{c \cdot h(d)})$$
for a given $d \in \mathbb{N} \cup \{0\}$.  Let $L \in \DTIME(n^c)$ be given.  By the hypothesis, $L \in \SPATIME(n^{c \cdot h(d)})$.  Consider an alternating Turing machine $M$
that decides $L$ within this time bound.  Consider a new machine $M'$ that decides $L$ where all of $M$'s exist and for all quantifiers are moved out front.
Therefore, the verification part of $M'$ is a deterministic $O(n^{c \cdot h(d)})$ time computation.

Now, we can take $p(n) = n^{c \cdot \frac{h(d + 1)}{\beta}}$ to get
$$\DTIME(n^{c \cdot h(d)}) \subseteq \DTIRE(n^{c \cdot h(d+1)}, \log(n)) \subseteq \SPATIME(n^{c \cdot h(d+1)}).$$
If we combine the alternating quantifiers with those from $M'$, then we get $$L \in \SPATIME(n^{\beta} + n^{c \cdot h(d + 1)}).$$
Since this works for any language $L \in \DTIME(n^c)$, we have $$\DTIME(n^c) \subseteq \SPATIME(n^{\beta} + n^{c \cdot h(d + 1)}).$$
Notice that an important reason why this works is because the combined witness length across all quantifiers is subpolynomial.
The base case and inductive step together imply that $\DTIME(n^c) \subseteq \SPATIME(n^{\beta} + n^{c \cdot h(d)})$ for all $d \in \mathbb{N} \cup \{0\}$.

By taking $d$ sufficiently large, we obtain $\DTIME(n^c) \subseteq \SPATIME(n^{\beta})$.  Since we also know that $\SPATIME(n^{\beta}) \subseteq \DSPACE(n^{\beta})$ \cite{alt1},
we have obtained $$\DTIME(n^c) \subseteq \SPATIME(n^{\beta}) \subseteq \DSPACE(n^{\beta})$$ for all $c \in \mathbb{N}$.  Therefore, $\PTIME \subseteq \SPATIME(n^{\beta}) \subseteq \DSPACE(n^{\beta})$.
\end{proof}

If in Proposition \ref{prop:atime1}, we instead make all of our existential quantifier blocks guess a single tape configuration,
then we have the following variation to the result from \cite{atime1,atime2,atime3,atime4}.

\begin{proposition}\label{prop:atime2}
Let functions $t(n)$ and $s(n)$ be given.  We have $$\NTISP(t(n), s(n)) \subseteq \ATIME^{a(n)}_{w(n)}(s(n)\log(t(n)))$$
for some $a(n)$ and $w(n)$ such that $a(n)$ is $O(\log(t(n)))$ and $w(n)$ is $O(s(n)\log(t(n)))$.
\end{proposition}

\begin{proof}
We proceed with a similar alternating simulation as in Proposition \ref{prop:atime1}, but we only guess one tape configuration
for each block of existential quantifiers.  The original computation is a sequence of $t(n)$ many tape configurations.
For this alternating simulation, each block of existential quantifiers will split a sequence of configurations in half.
A universal quantifier is then used to check both cases.  That is, the earlier configuration leads to the middle configuration
and the middle configuration leads to the later configuration.
Also, each block of existential quantifiers will
guess only $O(s(n))$ bits because the space bound is $O(s(n))$.  In total, this leads to
$O(\log(t(n)))$ alternations and witness length $O(s(n)\log(t(n)))$.  Furthermore, the time bound is $O(s(n)\log(t(n)))$ because
we only need to write down a sequence of middle configuration guesses and then finally check that one pair of configurations satisfies that the
earlier configuration leads to the later configuration in one step.
\end{proof}

By applying a repeated speed-up argument similar to Theorem \ref{hardnesshypo2} for a number of times that depends on the input size, we obtain the following corollary.

\begin{corollary}\label{hardnesshypo3}
The Hardness Hypothesis implies that $\PSPACE = \EXPTIME$.
\end{corollary}

\begin{proof}

Suppose that the Hardness Hypothesis holds for $k \in \mathbb{N}$, $\beta \geq 1$, and $\varepsilon > 0$.  Let $r = \frac{\beta}{\beta + \varepsilon}$.
Following the approach of Theorem \ref{hardnesshypo2} with the efficient alternating simulation from Proposition \ref{prop:atime2}, we obtain
$$\DTIME(2^n) \subseteq \DTIRE(2^{r \cdot n}, \log(n)) \subseteq \ATIME^{O(n)}_{O(n^2)}(2^{r \cdot n}).$$
Note that the second inclusion ($\subseteq$) follows by applying Proposition \ref{prop:atime2} with $t(n) = 2^{r \cdot n}$ and $s(n) = r \cdot n$.
By repeated application of these inclusions, we obtain $$\DTIME(2^n) \subseteq \ATIME^{O(d \cdot n)}_{O(d \cdot n^2)}(2^{r^d \cdot n})$$ for all $d \in \mathbb{N}$.
Furthermore, we proceed by carrying out this construction where $d$ depends on the input length $n$.
In particular, taking $d = \log_{1 / r}(n)$ we obtain $$\DTIME(2^n) \subseteq \ATIME^{O(n \log(n))}_{O(n^2 \log(n))}(n^{\beta} + n^{2} \log(n)) \subseteq \DSPACE(n^{\beta} + n^2 \log(n)).$$
It follows that $\EXPTIME \subseteq \PSPACE$.

It remains to explain how we carry out the construction where $d$ depends on $n$.
First, we observe that by the hardness hypothesis along with Proposition \ref{prop:atime2} we have
$$\DTIME(n^{\beta + \varepsilon} \log(n)) \subseteq \DTIRE(n^{\beta} \log(n), \log(n)) \subseteq \ATIME^{O(\log(n))}_{O(\log^2(n))}(n^{\beta} \log(n)).$$
Next, consider the language $U = \{ \, \texttt{enc}(M)\#s \,\, \vert \,\, \text{M accepts s in at most} \, \vert s \vert^{\beta + \varepsilon} \, \text{steps} \, \}$.  In other words,
$U$ consists of encoded Turing machine and input string pairs where the machine accepts the string within $n^{\beta}$ steps.  By \cite{time2}, we know
that $U \in \DTIME(n^{\beta + \varepsilon} \log(n))$.  Therefore, $$U \in \ATIME^{O(\log(n))}_{O(\log^2(n))}(n^{\beta} \log(n)).$$
Consider an alternating machine $\mathcal{M}$ that decides $U$ within these time, alternation, and witness length bounds.
Finally, let a language $L \in \DTIME(2^n)$ be given.  We recursively define a sequence $\{ M_d \}_{d \in \mathbb{N}}$ of alternating Turing
machines that decide $L$.
We define $M_1$ by taking the Turing machine that decides $L$ deterministically in $O(2^n)$ time, applying a padding phase, and plugging an encoded machine
with an input string into $\mathcal{M}$.
Given $M_d$, we define $M_{d+1}$ as follows.  We pull out the quantifiers for $M_d$ to get the inner verifier.
We next apply a padding phase and use the inner verifier to plug an encoded machine with an input string into $\mathcal{M}$.
We combine the pulled out quantifiers with the quantifiers from $\mathcal{M}$ to obtain $M_{d+1}$.

Now, given $d$ as input, we can compute $M_d$ efficiently and the size of $M_d$ only grows polynomially with $d$.
Therefore, we can carry out the construction as described obtaining $\EXPTIME \subseteq \PSPACE$.
It is worth noting that the logarithmic factor added from using $U$ and $\mathcal{M}$ does not have a significant effect
on the resource bounds in the construction because it gets reduced with each iteration of the speed-up argument.
\end{proof}

\section{Conclusion}

In Theorem \ref{conditionallower1}, we introduced an improved conditional time complexity lower bound for the $k$-$\dfaint$ problems.
Next, in Corollary \ref{timelower3}, we demonstrated that $\dfaint$ is not solvable in $O(n^{2 - \varepsilon})$ time unconditionally for all $\varepsilon > 0$.
Then, in Corollary \ref{timelower4}, we improved this lower bound to include the case when a limited amount of non-determinism is permitted.
Finally, in Corollary \ref{hardnesshypo3}, we showed that if any of the $k$-$\dfaint$ problems are hard for polynomial time classes
under sufficiently fast reductions, then $\PSPACE = \EXPTIME$.
In particular, we have that the Hardness Hypothesis (from Section \ref{sec:hardhypo}) implies that $\PSPACE = \EXPTIME$.

From these results, we conclude that the $\dfaint$ problem is not only computationally hard for space bounded computations, but also for strongly subquadratic time.
This offers us hope to show a definitive time complexity lower bound for one of the $k$-$\dfaint$ problems for which we currently only have conditional time lower bounds.
It also provides a general strategy for showing unconditional time complexity lower bounds for other $\PSPACE$-complete problems such as was shown for non-emptiness for two-way DFA's in Corollary \ref{timelower5}.
Lastly, the results of Section \ref{sec:hardhypo} suggest that if the $k$-$\dfaint$ problems are really as hard as they appear to be, then it would be possible to improve the recent breakthrough on
the space efficient simulation of time bounded computations from \cite{williams2025}.  In particular, the hardness of $k$-$\dfaint$ for
fixed polynomial time complexity classes (as considered in Proposition \ref{hardnesshypo4}, Theorem \ref{hardnesshypo2}, and Corollary \ref{hardnesshypo3}) would imply that all polynomial time problems are solvable in fixed polynomial
deterministic space.

\bibliography{main}

\end{document}